\documentclass[a4size,11pt,fleqn]{article}
\usepackage[T1]{fontenc}
\usepackage[utf8]{inputenc}
\usepackage{amsmath,amsthm,amssymb,amsfonts}
\usepackage{geometry}
\usepackage{graphicx}
\usepackage{tikz}
\usepackage{fancyvrb}

\usepackage{xcolor}
\definecolor{linkColor}{RGB}{0, 128, 128}
\definecolor{citeColor}{RGB}{0, 112, 64}
\definecolor{urlColor}{RGB}{120, 0, 120}
\usepackage{hyperref}
\hypersetup{
    colorlinks=true,
    linkcolor=linkColor,
    citecolor=citeColor,      
    urlcolor=urlColor,
}
\usepackage{url}

\newcommand{\regN}{\mathsf{N}}
\newcommand{\regW}{\mathsf{W}}
\newcommand{\regMn}{\mathsf{M}}

\newcommand{\cH}{\mathcal{H}}

\newcommand{\bZ}{\mathbb{Z}}

\theoremstyle{plain}
\newtheorem{thm}{Theorem}
\newtheorem{lem}[thm]{Lemma}

\newtheorem{hyp}[thm]{Hypothesis}
\newtheorem{clm}[thm]{Claim}

\theoremstyle{definition}
\newtheorem{defn}{Definition}

\renewcommand{\>}{\rangle}
\newcommand{\<}{\langle}

\newcommand{\gapFr}[2]{{{}^{#1}\hspace{-1pt}\{\!\!\{#2\}\!\!\}}}
\newcommand{\fr}[1]{{\{\!\!\{#1\}\!\!\}}}
\newcommand{\fixFr}[2]{{[\![#1..#2]\!]}}

\newcommand{\iniSt}{{\pmb{0}}} 
\newcommand{\Nodes}{{\mathcal{V}}} 

\title{Non-trivial lower bound for 3-coloring \\ the ring
in the quantum LOCAL model
}
\author{Fran\c{c}ois Le Gall and Ansis Rosmanis
\\
\normalsize  Graduate School of Mathematics \\ 
\normalsize Nagoya University
}
\date{December 6, 2022}

\begin{document}

\maketitle

\begin{abstract}
We consider the LOCAL model of distributed computing, where in a single round of communication each node can send to each of its neighbors a message of an arbitrary size. It is know that, classically, the round complexity of 3-coloring an $n$-node ring is $\Theta(\log^*\!n)$. In the case where communication is quantum, only trivial bounds were known: at least some communication must take place.

We study distributed algorithms for coloring the ring that perform only a single round of one-way communication. Classically, such limited communication is already known to reduce the number of required colors from $\Theta(n)$, when there is no communication, to $\Theta(\log n)$. In this work, we show that the probability of any quantum single-round one-way distributed algorithm to output a proper $3$-coloring is exponentially small in $n$.
\end{abstract}

\section{Introduction}

Graph coloring is one of the most fundamental tasks considered in distributed computing (see~\cite{BarenboimElkinBook} for a recent introduction). The coloring of a ring is particularly well studied and understood. While it is easy to see that 2-coloring an $n$-node ring requires exactly $n/2$ rounds for an even $n$ and is impossible for an odd $n$, things are much more interesting for 3-coloring.
Cole and Vishkin presented a $\frac{1}{2}\log^*\!n+\mathrm{O}(1)$ round algorithm for 3-coloring~\cite{COLE198632,Goldberg88coloring}.
Linial~\cite{Linial92local} showed the Cole--Vishkin algorithm to be optimal for deterministic algorithms, and then Naor~\cite{NaorColoringLowerBound} extended that lower bound to randomized algorithms.

The complexity of graph coloring is much less understood in the quantum setting due to the difficulty of characterizing the power of quantum entanglement. While quantum entanglement enables non-local correlations, it by itself does not enable communication. This property is known as non-signaling, and Gavoille, Kosowski, and Markiewicz used it to show that quantum distributed algorithms for 2-coloring require at least $n/4$ rounds of communication~\cite{GKM2009}. As for 3-coloring, up to the best of our knowledge, prior to the current work no non-trivial bounds on the round complexity were known. Namely, it was clear that at least some communication must take place, because a non-communicating quantum node is no different from a non-communicating classical node, but no more than that was known.

We consider the LOCAL model introduced by Linial~\cite{Linial92local,peleg2000distributed}, where the bandwidth of communication links is unlimited. We additionally assume that the ring is directed and each node knows which of its two communication ports is for communicating with its successor, and which for its predecessor. One can consider the standard, two-way model, where the communication can go both along and against the direction of a directed link, and the one-way model, where the messages can be sent only in the direction of the links (see Figure~\ref{fig:2way1way}).

\DeclareRobustCommand{\myarrow}{\tikz{\draw [-latex](0,0) -- (0.7,0);\draw [opacity=0.2,white](0,-0.1) -- (0.7,0.1);}} 

\begin{figure}[!h]
\begin{centering}
\includegraphics[width=4cm]{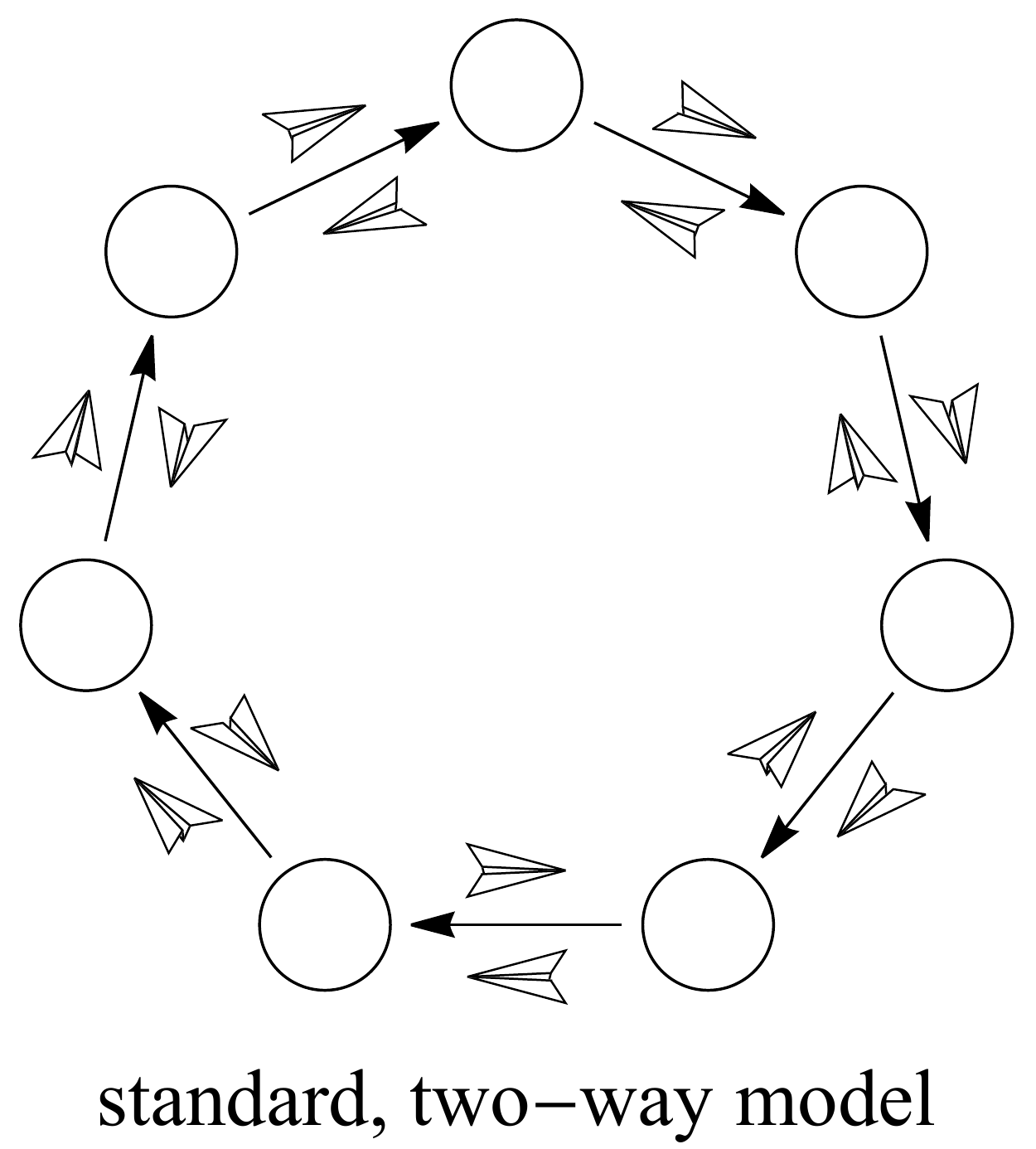}
\qquad\qquad
\includegraphics[width=4cm]{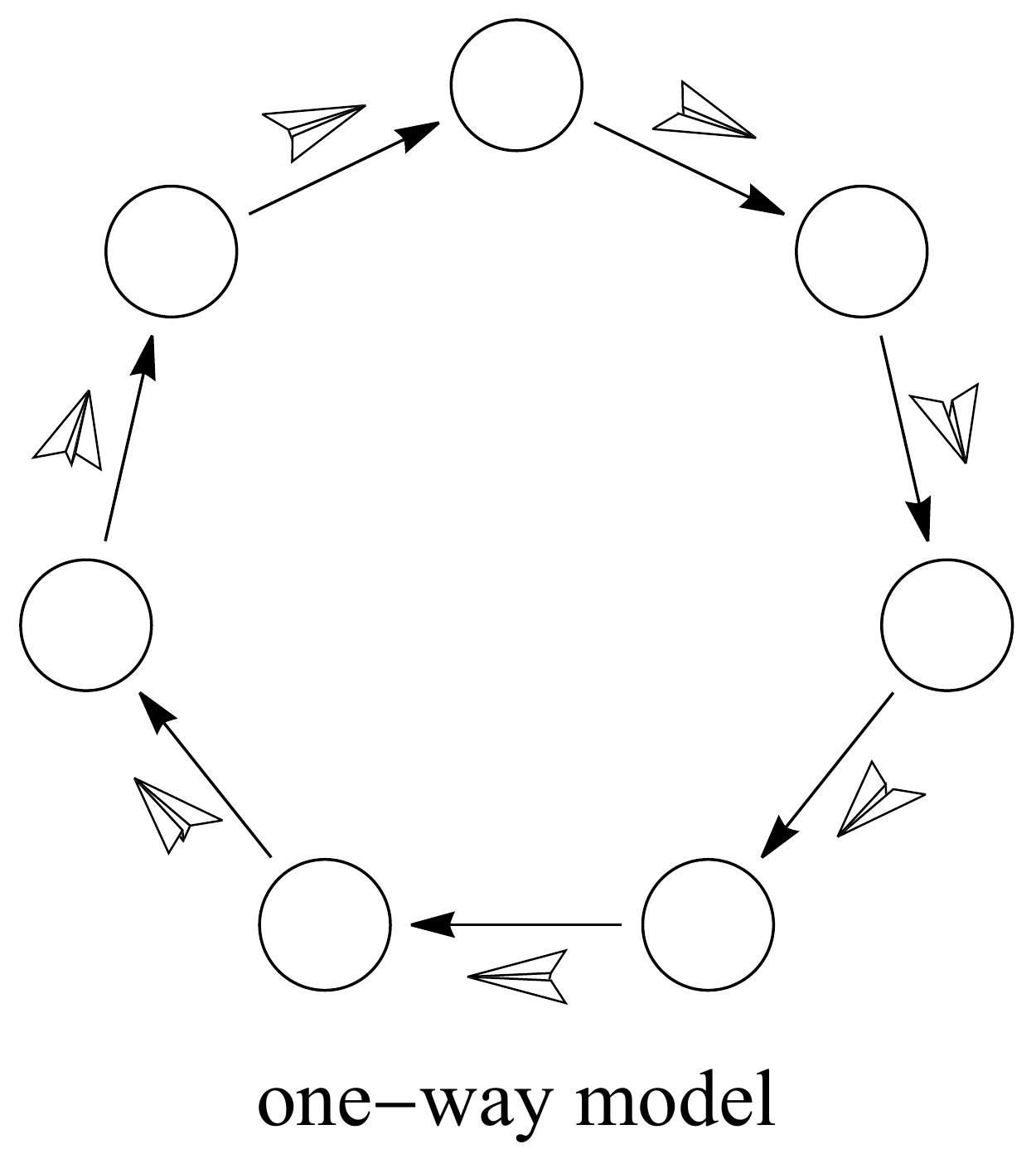}
\caption{\small
The standard and the one-way LOCAL models on the ring. The direction of links is indicated by arrows (\myarrow),
 while the direction of communication is indicated by the direction of the paper plane symbols.}
\label{fig:2way1way}
\end{centering}
\end{figure}

We can observe that, due to node's ability to forward its entire memory to its successor, every locally checkable labeling problem%
\footnote{Essentially, a problem for which the correctness of a solution can be checked in a constant number of rounds~\cite{NaorStockmeyer:LCL}. Coloring is clearly locally checkable.}
 that can be optimally solved in $r$ rounds in the one-way model can be solved in $\lceil r/2\rceil$ rounds in the standard model, and no faster. Thus, studying the one-way model can provide a more fine-grained characterization of the complexity of the problem. For example, $r$ rounds of the standard model enable $\mathrm{O}(\log^{(2r)}\!n)$-coloring, while $r$ rounds of the one-way model enable $\mathrm{O}(\log^{(r)}\!n)$-coloring.
This fine-graining also allows us to prove the first non-trivial hardness result for quantum distributed algorithms for $3$-coloring:

\begin{thm}
\label{thm:expBound}
The probability that a single-round one-way quantum distributed algorithm properly 3-colors an $n$-node ring is exponentially small in $n$; in particular, it is at most $(\frac{11}{15})^{\lfloor n/10\rfloor}$.
\end{thm}

The inspiration for our result is the observation that the uniform distribution over all proper colorings exhibits some very long-range correlations, which we detail below in Section~\ref{sec:introUniform}. We further observe that all proper colorings exhibit a certain correlation between the distance between two randomly chosen nodes and whether these nodes have the same color. We sketch this latter observation for an 11-node ring as an example in Section~\ref{intro:n11}, and we note that this observation leads to a weaker version of Theorem~\ref{thm:expBound} having the bound $(\frac{1382}{1383})^{\lfloor n/16\rfloor}$ instead of $(\frac{11}{15})^{\lfloor n/10\rfloor}$. 
Improving the base of the exponent from $(\frac{1382}{1383})^{1/16}$ to $(\frac{11}{15})^{1/10}$ requires a more thorough inspection of random colorings that may arise from quantum distributed algorithms. Here, in Section~\ref{sec:introLP}, we briefly outline that approach, as well as describe how it might help to address the case of multiple-round distributed algorithms in the future.

\subsection{Uniform coloring requires global communication}
\label{sec:introUniform}

As discussed in the works by Gavoille, Kosowski, and Markiewicz~\cite{GKM2009} and Arfaoui and Fraigniaud \cite{AF2color}, the non-signaling arguments, which are used to prove the hardness of 2-coloring in the quantum case, do not seem to be able to generalize the classical $\Omega(\log^*\!n)$ lower bound for 3-coloring~\cite{Linial92local,NaorColoringLowerBound} to the quantum case. To formally confirm the correctness of such an intuition, one would have to answer the following question in the affirmative, and trying to do exactly that was the starting point of our research. The raised question is:

\bigskip

\noindent
\emph{Can one construct a probability distribution over proper $3$-colorings of the $n$-node ring that is non-signaling beyond distance $o(\log^*\!n)$?}%
\footnote{Informally, a distribution is non-signaling beyond distance $r$ if, for any collection of ring segments, the marginal distributions over colorings of these segments do not depend on the mutual distances between the segments, as long as all these distances are more than $r$.}

\bigskip

\noindent
Our first attempt at such a construction was the uniform distribution over all proper 3-colorings.
However, as the following argument shows, the uniform distribution is highly non-local: if one knows exactly the probability that two given nodes have the same color, one can determine the distance $d\in\{1,2,\ldots,\lfloor n/2\rfloor\}$ between them.

Consider a line segment of length $\ell$, consisting of $\ell+1$ nodes connected by $\ell$ links. It is easy to see that there are $3\cdot 2^\ell$ proper colorings of this line segment. Let $a_\ell$ and $b_\ell$ denote the number of such proper colorings that, respectively, have and have not both endpoints colored in the same color. By solving a simple recurrence relation $a_0 = 3$, $b_0= 0$, and $a_\ell = b_{\ell-1}$ and $b_\ell = 2a_{\ell-1} + b_{\ell-1}$ for $\ell\ge 1$, one can see that
\[
a_\ell = 2^\ell+2(-1)^\ell
\quad\text{and}\quad
b_\ell = 2(2^\ell-(-1)^\ell).
\]
A properly colored $n$-node ring can be thought of as a properly colored line segment of length $n$ with both endpoints being colored in the same color. Hence, there are $a_n$ proper colorings of such a ring, and the probability that two given nodes at distance $d$ have the same color is
\[
\frac{a_da_{n-d}}{3a_n}
=
\frac{(2^d+2(-1)^d)(2^{n-d}+2(-1)^{n-d})}{3(2^n+2(-1)^n)}
=
\frac{1}{3}+ \frac{2(-1)^n}{3}\cdot
\frac{1 + (-2)^{d}+(-2)^{n-d}}{2^n+2(-1)^n},
\]
which is distinct for every $d$.

As a result, using non-signaling arguments along the same lines as in~\cite{GKM2009}, it can be seen that achieving the uniform distribution over all proper 3-colorings is no faster than 2-coloring. It therefore also follows that the distribution over proper 3-colorings that results from uniformly at random assigning unique labels to nodes and then running the Cole--Vishkin algorithm is itself not uniform. We show in Section~\ref{sec:IndAndNSDistr} that distributions that are non-signaling beyond a given distance $r$ must satisfy certain linear constraints, which can then be incorporated in a linear program to bound the maximum probability for quantum distributed algorithms to properly 3-color the ring.

\subsection{Impossibility of perfect coloring}
\label{intro:n11}

As we discovered, looking for correlations between distances among nodes and whether those nodes have the same color not only establishes the globality of the uniform distribution, but also leads to certain no-go results for arbitrary distributions. Our proof that the probability of a one-round algorithm to output a proper 3-coloring is at most $(\frac{1382}{1383})^{\lfloor n/16\rfloor}$ is essentially a generalization of the following observation, which right away implies that it is impossible for a quantum single-round one-way distributed algorithm to perfectly  color (i.e., without any error) an $11$-node ring.

\bigskip

Given a coloring of the ring of $11$ nodes, consider two experiments: Experiment $1$ and Experiment $2$. For $e\in\{1,2\}$, Experiment $e$ starts by first randomly choosing the distance $d_e$, where
\[
d_1= \begin{cases}
 2 & \text{w.p. }\frac{30}{41}, \\
 3 & \text{w.p. }\frac{11}{41}, \\
 \end{cases}
\quad\text{and}\quad
d_2= \begin{cases}
 4 & \text{w.p. }\frac{14}{41}, \\
 5 & \text{w.p. }\frac{27}{41}. \\
 \end{cases}
\]
Then the experiment chooses uniformly at random one of $11$ nodes, and compares its color with the node distance $d_e$ ahead of if, outputting \texttt{Match} if the colors match and \texttt{Differ} if they do not.

Let us consider running these experiments on proper colorings. There are $a_{11}=2046$ such colorings, so a simple computer program (see Appendix~\ref{app:code11}) can quickly check for all of them what is the probability of Experiments 1 and 2 each returning \texttt{Match}. As it turns out, for every proper coloring, the probability of Experiment 1 returning \texttt{Match} is 
 at least by $\frac{1}{451}$
 greater that the probability of Experiment 2 returning \texttt{Match}.

Now suppose we had a quantum single-round one-way distributed algorithm that perfectly  $3$-colors an $11$-node ring, producing some probability distribution over proper colorings. 
Due to non-causality, every two non-adjacent nodes cannot tell what is the distance between them. Hence, the probability of them returning the same color must be independent from the distance between them.
However, the observation above implies that it is impossible.
More quantitatively, in Appendix~\ref{sec:boundsN11N22} we show that it implies that the success probability is at most $\frac{451}{452}$. If we additionally analyze Experiments 1 and 2 on improper colorings, we can improve this bound to $\frac{244}{245}$.

\bigskip

To prove the exponential bound $(\frac{1382}{1383})^{\lfloor n/16\rfloor}$, we have to use a similar observation for line segments. In addition, we have to consider a more limited set of distributions than non-signaling distribution, which nonetheless includes all the distributions resulting from quantum distributed algorithms.

\subsection{Linear program for bounding success probability}
\label{sec:introLP}

Probability distributions resulting from quantum $r$-round one-way distributed algorithms are non-signaling beyond distance $r$. 
Hence, the maximum among success probabilities of non-signaling distributions---the total probability weight they place on proper colorings---upper-bounds the maximum success probability of quantum distributed algorithms.

To obtain stricter bounds on the success probability than the ones obtained along the lines of comparing Experiments 1 and 2, we can exploit the fact that non-signaling distributions have to satisfy much more constrains than that the probability of two distant nodes having the same color must be independent from the distance between the nodes.
In particular, we show how to express the maximum success probability of non-signaling distributions beyond distance $r$ as a linear program. As a result, for example, for $n=11$ and a single round, $r=1$, we show that the success probability is at most $\frac{32}{63}$, highly improving upon the bound $\frac{244}{245}$ described above.

To prove Theorem~\ref{thm:expBound}, we further have to consider line segments and non-signaling distributions over their colorings. The maximum among their success probabilities can also be expressed as a linear program.

While it might be possible that this linear program shows impossibility of perfect coloring for quantum one-way distributed algorithms of more than one round, computationally solving the linear program of size large enough to witness this impossibility for $r$ much larger than $1$ is infeasible: one has to consider the number of nodes that is at least the inverse function of $\log^*$ in $r$, which is astronomically large.

Nevertheless, we hope that a closer inspection of impossibility to perfectly color line segments---in particular, the inspection of optimal primal and dual solutions of the linear program---could lead to a technique similar to the round elimination~\cite{Brand2019RoundElem}, which is the current go-to theorem for showing the $\Omega(\log^*\!n)$ lower bound for classical randomized distributed algorithms (see \cite{SuomelaDA2020} for a pedagogical description of the technique).

\subsection{Organization of the paper}

The paper is organized as follows. In Section~\ref{sec:prelim}, we formalize some aspects of the problem of coloring the ring and we informally introduce the model of quantum distributed algorithms, leaving the formal definition of the model to the appendix. In Section~\ref{sec:IndAndNSDistr}, we define independent and non-signaling random colorings, and provide a framework for obtaining exponential bounds on the success probability. 
In Section~\ref{sec:SuccessProbLP}, we study the linear program for the maximum success probability of non-signaling random colorings of the ring and line segments. Using Wolfram Mathematica, we symbolically solve the linear program for certain cases when the number of nodes is small. These solutions together with the framework introduced in Section~\ref{sec:IndAndNSDistr} yield our main result, Theorem~\ref{thm:expBound}.
Independent random colorings are more constrained that non-signaling colorings, and, in Section~\ref{sec:independent:n4}, we use them to prove that even a 4-node ring cannot be colored perfectly in a single one-way round. However, the optimal bounds on the success probability of independent random colorings do not seem to be expressible as linear programs, prohibiting us from obtaining similar results for larger number of nodes. Finally, in Section~\ref{sec:FutureWork}, we talk about potential approaches towards proving impossibility results beyond a single round of communication and other future research.

To make the paper more accessible to readers without a background in quantum computing, aside from a brief description of quantum distributed algorithms in Section~\ref{sec:prelim}, we keep the main body of the text free from details of quantum computation. We describe the model of quantum distributed computation in more detail in Appendix~\ref{app:quantum}, where we also prove Lemma~\ref{lem:QuantumNonSignal} stating that the random colorings returned by such computation are independent at distances beyond the number of rounds of communication. 
In Appendix~\ref{sec:No3Coloring}, we observe certain correlations between the distances among the nodes and whether they have the same color  in proper colorings of the ring and line segments. These observations together with the framework introduced in Section~\ref{sec:IndAndNSDistr} yield the bound $(\frac{1382}{1383})^{\lfloor n/16\rfloor}$, which is weaker than Theorem~\ref{thm:expBound}, but the observations themselves might be of an independent interest in combinatorics.
In Appendix~\ref{app:code11}, we provide and annotate the Wolfram Mathematica code that we have used for analyzing the example of two experiments for the 11-node ring described in Section~\ref{intro:n11}. This and the rest of the Wolfram Mathematica code used for our calculations is available online~\cite{code3Color}.

\section{Preliminaries}
\label{sec:prelim}

Let us start by introducing terminology and mathematical formalism that we use when discussing the directed ring.

\begin{defn}
An \emph{$n$-node directed ring} is a directed degree-$1$ graph with the set of vertices $\Nodes:=\mathbb{Z}_n=\{0,1,\ldots,n-1\}$ and the set of arcs $\{(v,v+1)\colon v\in \Nodes\}$. 
We call $v\in \Nodes$ a \emph{node} and we also refer to $v$ as the \emph{number} of that node.  We call arcs of the graph \emph{(communication)} \emph{links}. We say that $v-1$ \emph{precedes} $v$, $v+1$ \emph{succeeds} or \emph{follows} $v$, and that they are both \emph{adjacent} to $v$. We also call $v+1$ the \emph{next} node of $v$.
\end{defn}

\noindent
In the above definition and throughout the paper, for a node $v$ and any integer $k$, $v+k$ denotes the node number $v+k\mod n$.
Given $v\in \Nodes$ and $D\subseteq\mathbb{Z}_n$, let
\( v+D := \{v+d\colon d\in D\}\subseteq \Nodes.\)

Let $\Sigma$ be some fixed finite set that we call the \emph{set of colors}. We call the function $\phi\colon \Nodes\rightarrow \Sigma$ a \emph{coloring}, and we say that $\phi$ is \emph{proper} if $\phi(v)\ne\phi(v+1)$ for all $v\in \Nodes$. For brevity, we may also denote the set of all $|\Sigma|^n$ colorings by $\Sigma^\Nodes$. We also consider colorings of segments of the ring, and their properness is defined analogously: every pair of adjacent nodes must have different colors.

For a tuple $(a_1,\ldots,a_\ell)$ of distinct elements, let $\mathsf{set}(a_1,\ldots,a_\ell):=\{a_1,\ldots,a_\ell\}$.

\paragraph{Computational model.}

Here we describe the computational model of quantum distributed algorithm on the directed ring. We leave the precise definition of the model to Appendix~\ref{app:quantum}, where we provide the proof of Lemma~\ref{lem:QuantumNonSignal}, which is a slight adaptation%
\footnote{Here we take into account the one-wayness of the communication, and we also consider a special case of non-signaling random colorings that we call independent colorings.} 
 of the folklore fact that random colorings resulting from quantum distributed algorithms are non-signaling. Aside from establishing Lemma~\ref{lem:QuantumNonSignal}, no other aspects of quantum computation are used.

Each node operates its own quantum processor, executing exactly the same local algorithm as all the other nodes. In addition, every node can receive quantum messages from the node preceding it and send quantum messages to the node succeeding it. This circular communication proceeds in synchronous rounds. We are considering the LOCAL model~\cite{Linial92local}, where there is no bandwidth limitations on the size of messages.
At the end of the computation, after $r$ rounds, each node outputs a color from $\Sigma$, and we treat the corresponding mapping $\Phi$ from $\Nodes$ to $\Sigma$ as the coloring.
Due to the probabilistic nature of quantum computation, the output $\Phi$ is probabilistic, and we use the capital letter $\Phi$ instead of $\phi$ to emphasize that we think of coloring $\Phi$ as a random variable taking values in $\Sigma^\Nodes$.

\section{Independent and Non-Signaling Random Colorings}
\label{sec:IndAndNSDistr}

Arfaoui and Fraigniaud~\cite{AF2color} used the non-signaling property of quantum computation when analyzing, along the lines of Ref.~\cite{GKM2009}, the hardness of $2$-coloring the ring for quantum distributed algorithms. Here we define \emph{non-signaling colorings} and their special case that we call \emph{independent colorings}. Our definition of non-signaling colorings slightly differs from the definition of non-signaling distributions by Arfaoui and Fraigniaud. First of all, our definition is specific to the ring graph, not a general graph, and we have also made adaptations arising from the communication being one-way. In addition, in the distributed model that we consider, nodes do not receive unique identifiers at the beginning of the computation. However, aside from adaptations due to the communication being one-way, our definition would match that of Arfaoui and Fraigniaud if they assigned the unique identifiers to the nodes uniformly at random.

We consider colorings \(\Phi\) of the whole ring and colorings $\Psi$ of segments of the ring, both of which are random variables. Definitions and lemmas established here will let us prove our main result, Theorem~\ref{thm:expBound}, in Section~\ref{sec:SuccessProbLP}.

\subsection{Frames, tableaux, and their collections}

Let $s\in\{1,2,\ldots,n-1\}$.
We call the intervals of integers
\[
\fr{s}:=(0,1,\ldots,s-1)
\qquad\text{and}\qquad
\gapFr{r}{s}:=(-r,-r+1,\ldots,s-1),
\]
a \emph{sliding frame} and a \emph{gapped sliding frame}, respectively.
See Figure~\ref{fig:FramesOnRing} for illustrations.
We refer to $s$ as the \emph{length} of the frame and to $r$ as the length (or width) of the \emph{gap}.
For nodes $v$ and $v'=v+s-1$, we call the interval of nodes
\[
\fixFr{v}{v'}:=(v,v+1,\ldots,v')
\]
a \emph{fixed frame} of length $s$. Given a vertex $v$, let $v+\fr{s}:=\fixFr{v}{v+s-1}$ and $v+\gapFr{r}{s}:=\fixFr{v-r}{v+s-1}$ 
(here we assume $s+r\le n-1$).

Intuitively, the goal of the above definition is to ensure the following regarding the operation of $r$-round distributed algorithms. If $v+\gapFr{r}{s}$ and $v'+\gapFr{r}{s'}$ do not overlap, then colors output by $v+\fr{s}$ and $v'+\fr{s'}$ should be independent. We can think of a gap as a buffer that keeps (non-gapped) sliding frames sufficiently apart.

\begin{figure}[!h]
\begin{centering}
\includegraphics[width=6cm]{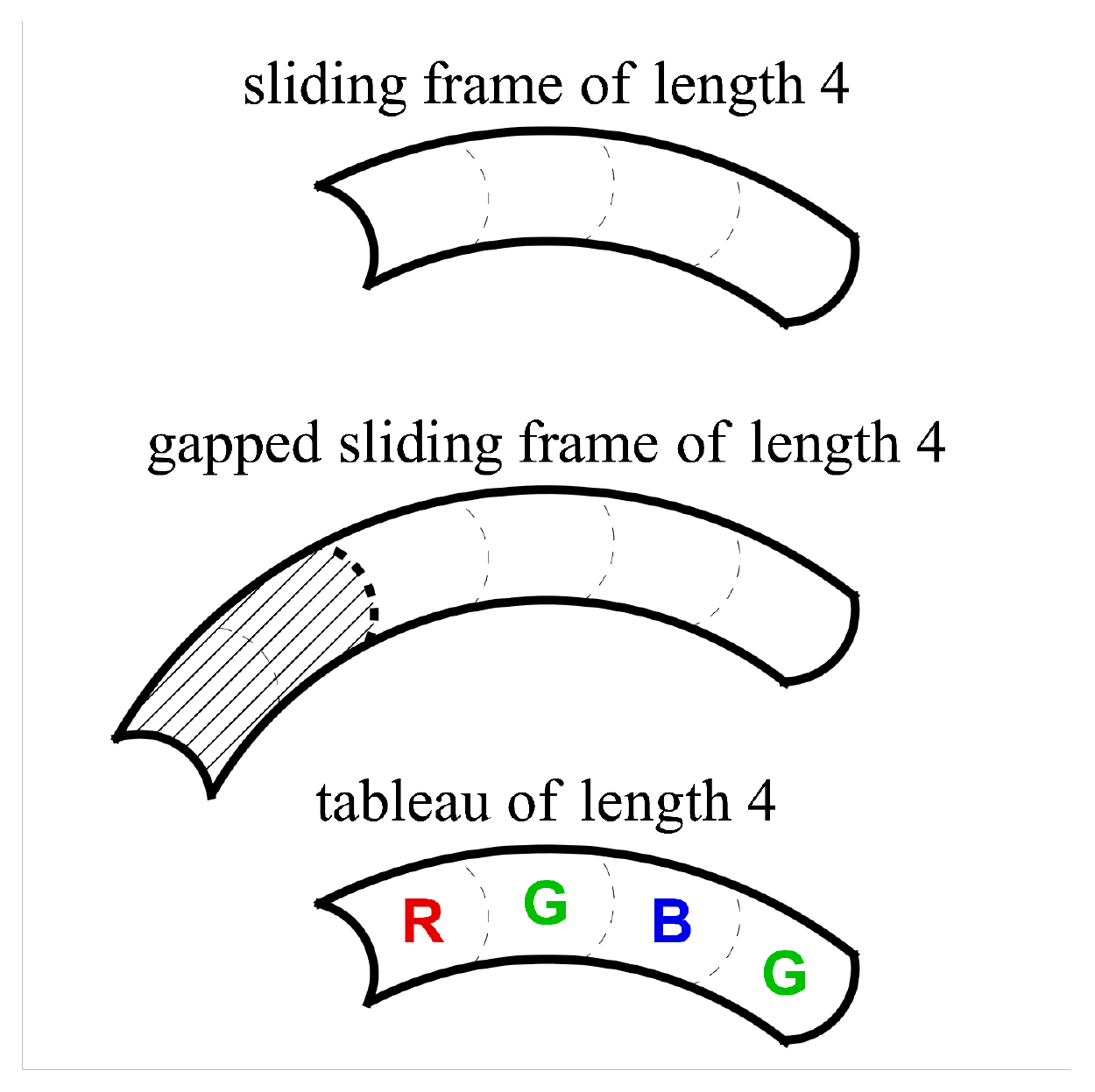}
\includegraphics[width=6cm]{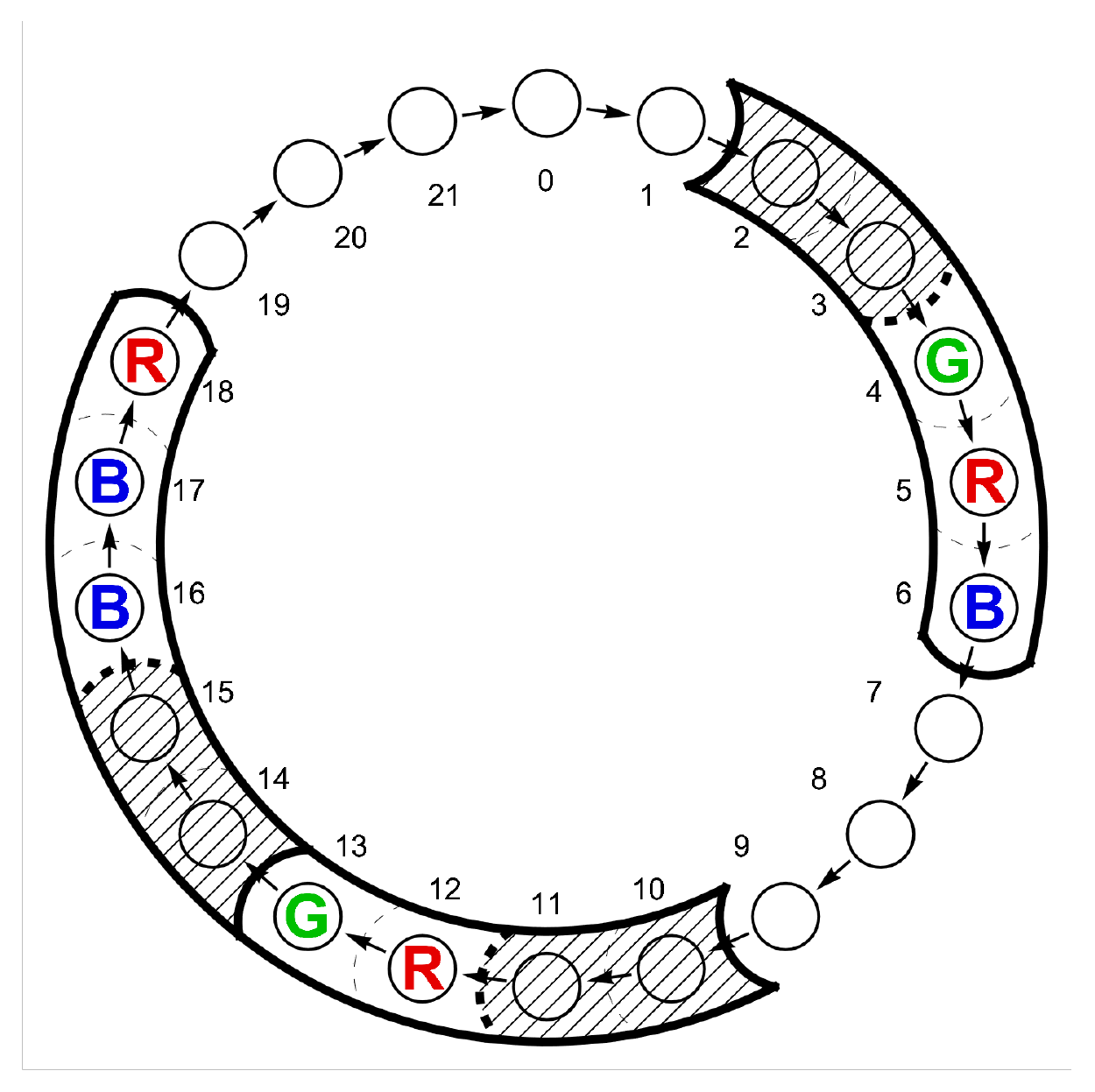}
\caption{\small
On the left, a sliding frame, a gap-$2$ gapped sliding frame, and a tableau, all of length $4$, are illustrated. Here $\Sigma=\{R,G,B\}$. We can think of gapped sliding frames as of something that we can slide circularly on top of the ring (including jumping over one another), and we need to place them so that they do not overlap. On the left, the ring of length $n=22$, on which the collection $F=(\fr{3},\fr{3},\fr{2})$ of frames is placed by the gap-$2$ placement $\omega=(4,16,12)$. Assuming there is some underlying coloring $\phi$ of the ring, we have revealed the colors corresponding to the sliding frames, the $(3,3,2)$-collection of tableaux in this case being $((G,R,B),(B,B,R),(R,G))$.}
\label{fig:FramesOnRing}
\end{centering}
\end{figure}

 Given a coloring $\phi$, let $\phi(\fixFr{v}{v'})$ be short for $(\phi(v),\phi(v+1),\ldots,\phi(v'))\in\Sigma^s$, and we call such a tuple of colors a \emph{tableau}.
We note that we \emph{cannot} think of $\phi(\fixFr{v}{v'})$ as the partial function that is the restriction of $\phi$ to the subdomain $\mathsf{set}\fixFr{v}{v'}$. 
In particular, we may have $\phi(\fixFr{v}{v'})=\phi(\fixFr{v+\delta}{v'+\delta})$ for some ``shift'' $\delta\in\bZ$.

\paragraph{Collections of frames and tableaux.}

We call a $t$-tuple $(\fr{s_1},\ldots,\fr{s_t})$ a \emph{collection of sliding frames}.
Given a collection of sliding frames $F=(\fr{s_1},\ldots,\fr{s_t})$, we say that $F$ is \emph{gap-$r$-placable} if $\sum_{j=1}^t (s_j+r)\le n$. We say that a tuple  $\omega\in \Nodes^t$ is a \emph{gap-$r$ placement} of $F$ if for all distinct $j,j'\in\{1,\ldots,t\}$ we have
\[
\mathsf{set}\big(\omega_j+\gapFr{r}{s_j}\big) \cap \mathsf{set}\big(\omega_{j'}+\gapFr{r}{s_{j'}}\big) = \emptyset.
\]

Given a collection of sliding frames $F=(\fr{s_1},\ldots,\fr{s_t})$ and a gap-$r$ placement $\omega$ of $F$, let
\[
\omega+F:=
\big(
\omega_1+\fr{s_1},\,
\ldots,\,
\omega_t+\fr{s_t}
\big).
\]
For a coloring $\phi$, let
\[
\phi(\omega+F):=
\big(
\phi(\omega_1+\fr{s_1}),\,
\ldots,\,
\phi(\omega_t+\fr{s_t})
\big),
\]
which is a $t$ tuple of tableaux, with $j$-th tableau having length $s_j$.
We call such a tuple an $(s_1,\ldots,s_t)$-collection of tableaux.
We note that the order of tableaux in a collection is fixed. For example, given $\Sigma=\{R,G,B\}$, collections $((R,G),(G,B))$ and $((G,B),(R,G))$ are not the same. 

\subsection{Colorings of the ring}

We consider random variables $\Phi$ whose values are colorings $\phi\colon \Nodes\rightarrow\Sigma$, and we refer to $\Phi$ as a \emph{random coloring}, or we might still just simply refer to it as a coloring. Let $\mathsf{Prev}\colon \Nodes\rightarrow \Nodes\colon v\rightarrow v-1$.
We call a random coloring $\Phi$ \emph{cyclic} if 
\[
\Pr[\Phi = \phi\circ\mathsf{Prev}] = \Pr[\Phi = \phi],
\]
where $\circ$ denotes the composition.

Given a random coloring $\Phi$, we define the random variable $\Phi({\omega+F})$ accordingly:
\[
\Pr[\Phi(\omega+F)=(\zeta_1,\ldots,\zeta_t)]
:=
\Pr[\Phi=\phi\text{ such that }\phi(\omega+F)=(\zeta_1,\ldots,\zeta_t)].
\]
We can think of $\Phi(\omega+F)$ inducing a marginal distribution over $(\zeta_1,\ldots,\zeta_t)$ in the natural way.

\begin{defn}
\label{def:independentColoring}
We say that a random coloring $\Phi$ is \emph{independent at distances beyond $r$} if it is cyclic and if for every $s\ge 1$ there exists a probability distribution $p_s$ over $\Sigma^s$ such that,
for every collection $F=(\fr{s_1},\ldots,\fr{s_t})$ of sliding frames, a gap-$r$ placement $\omega=(\omega_1,\ldots,\omega_t)$ of $F$, 
and $\zeta_j\in\Sigma^{s_j}$ for every $j$,
we have 
\[
\Pr[\Phi(\omega+F)=(\zeta_1,\ldots,\zeta_t)]
=
p_{s_1}(\zeta_1) \cdot \ldots \cdot p_{s_t}(\zeta_t).
\]
\end{defn}

\begin{defn}
\label{def:nonsignalingColoring}
We call a random coloring $\Phi$  \emph{non-signaling at distances beyond $r$} if it is cyclic and if, for every gap-$r$-placable collection of frames $F$, the marginal distribution induced by $\Phi(\omega+F)$ is the same for all gap-$r$ placements $\omega$ of $F$.
\end{defn} 

Note that, if $\Phi$ is non-signaling at distances beyond $r_1$, then it is also non-signaling at distances beyond $r_2$ for $r_2>r_1$.
The following directly follows from Definitions~\ref{def:independentColoring} and \ref{def:nonsignalingColoring}.

\begin{clm}
\label{clm:IndDoesntSignal}
A random coloring $\Phi$ that is independent at distances beyond $r$ is also non-signaling at distances beyond $r$.
\end{clm}

We note that the reverse of Claim~\ref{clm:IndDoesntSignal} is not true: there exist colorings that are non-signaling beyond some distance but that are not independent beyond that same distance. For example, suppose $|\Sigma|=n$, and consider $\Phi$ that is a uniform distribution over all $n!$ colorings $\phi$ that assign a unique color from $\Sigma$ to every node. $\Phi$ is non-signaling already beyond $r=0$, but it becomes independent only beyond $r\ge\lfloor 
n/2\rfloor$.%
\footnote{For $r\ge\lfloor n/2\rfloor$, only collections consisting of a single sliding frame are gap-$r$-placable.}

In Appendix~\ref{app:quantum}, we formalize the model of quantum distributed algorithms, adapted for one-way communication on the ring, and we show that the output distributions of quantum $r$-round one-way distributed algorithms are independent at distances beyond $r$.

\begin{lem}
\label{lem:QuantumNonSignal}
The coloring $\Phi$ produced by a quantum one-way $r$-round distributed algorithm is independent at distances beyond $r$.
\end{lem}

\subsection{Colorings of line segments and exponential bounds}

Let us consider colorings of line segments,
where by a line segment we think of some interval of nodes of the ring, not containing all of them (more precisely, at least $r$ of them).
Let $k\ge 1$, $v\in\Nodes$, and consider interval of vertices $v+\fr{k}=\fixFr{v}{v+k-1}$. We call a map $\psi\colon\mathsf{set}\fixFr{v}{v'}\rightarrow\Sigma$ a coloring of the line segment. Given a set of frames $F=(\fr{s_1},\ldots,\fr{s_t})$ we say that a gap-$r$ placement $\omega$ of $F$ \emph{respects} the interval $\fixFr{v}{v'}$ if
\[
\mathsf{set}(\omega_j+\fr{s})\subseteq \mathsf{set}\fixFr{v}{v'}
\]
for all $j$. Note that it is still allowed that  
$
\mathsf{set}(\omega_j+\gapFr{r}{s})\not\subseteq \mathsf{set}\fixFr{v}{v'},
$
for example, when $\omega_j=v$.

Along the same lines as Definitions~\ref{def:independentColoring} and \ref{def:nonsignalingColoring} for complete rings, we define independent and non-signaling colorings for line segments.

\begin{defn}
\label{def:independentColoringSegment}
We say that a random coloring $\Psi$ of a line segment $\fixFr{v}{v'}$ is \emph{independent at distances beyond $r$}  if for every $s\ge 1$ there exists a probability distribution $p_s$ over $\Sigma^s$ such that,
for every collection $F=(\fr{s_1},\ldots,\fr{s_t})$ of sliding frames, a gap-$r$ placement $\omega=(\omega_1,\ldots,\omega_t)$ of $F$ that respects $\fixFr{v}{v'}$, 
and $\zeta_j\in\Sigma^{s_j}$ for every $j$,
we have 
\[
\Pr[\Phi(\omega+F)=(\zeta_1,\ldots,\zeta_t)]
=
p_{s_1}(\zeta_1) \cdot \ldots \cdot p_{s_t}(\zeta_t).
\]
\end{defn}

\begin{defn}
\label{def:nonsignalingColoringLine}
We call a random coloring $\Psi$ of a line segment $\fixFr{v}{v'}$ \emph{non-signaling at distances beyond $r$} if, for every collection of frames $F$, the marginal distribution induced by $\Psi(\omega+F)$ is the same for all gap-$r$ placements $\omega$ of $F$ that respect $\fixFr{v}{v'}$.
\end{defn} 

\begin{lem}
\label{lem:SegmentsToExp}
Suppose there exist a length $k\ge 2$, a probability $q<1$, and a node $v\in\Nodes$ such that, for all random colorings $\Psi$ of the interval $\fixFr{v}{v+k-1}$ that are independent beyond distance $r$, the success probability of $\Psi$ is at most $q$.
Then for every coloring $\Phi$ of the whole ring that is independent beyond distance $r$ the success probability of $\Phi$ is at most $q^{\lfloor n/(k+r)\rfloor}$.
\end{lem}

\begin{proof}
Suppose $\Phi$ is independent (beyond distance $r$), and let $\Psi$ be its restriction to the domain $\mathsf{set}\fixFr{v}{v+k-1}$.
Since $\Phi$ is independent, first, so is $\Psi$, and, second, there is a probability distribution $p_k$ over $\psi\in\Sigma^k$ such that 
\[
\Pr[\Phi(\fixFr{v}{v+k-1})=\psi]
=
p_{k}(\psi).
\]
From the assumptions of the lemma,
\[
\sum_{\text{proper }\psi\in\Sigma^k } p_k(\psi) \le q.
\]

Let $m:=\lfloor n/(k+r)\rfloor$ for brevity.
Now consider the collection of frames $F=(\fr{k},\ldots,\fr{k})$, consisting of $m$ frames, each of length $k$, and consider their gap-$r$ placement
$\omega=(\omega_1,\ldots,\omega_m)$ where 
$\omega_j:=j(r+k)$.
Since $\Phi$ is independent, we have
\[
\Pr[\Phi(\omega+F)=(\zeta_1,\ldots,\zeta_m)]
=
p_{k}(\zeta_1) \cdot \ldots \cdot p_{k}(\zeta_m).
\]
As a result, the probability that for a random $(\zeta_1,\ldots,\zeta_m) = \Phi(\omega+F)$ we have that $\zeta_j$ is proper for all $j\in\{1,\ldots,m\}$ is at most $q^m$. Finally such a piece-wise properness is a prerequisite for the whole $\Phi$ to be proper, and the probability of $\Phi$ being proper can be no larger.
\end{proof}

\noindent
In Section~\ref{sec:nonsig9nodeLine}, we instantiate this lemma with $r=1$, $k=9$, and $q=11/15$, thus proving Theorem~\ref{thm:expBound}.

\section{Success Probability of Non-Signaling Colorings as a Linear Program}
\label{sec:SuccessProbLP}

If we inspect Definitions~\ref{def:nonsignalingColoring} and \ref{def:nonsignalingColoringLine} of non-signaling colorings of the ring and the line segment, respectively, we see that the maximum success probability among non-signaling colorings can be computed via a linear program.
In this section, we provide details of this claim by stating the two linear programs.
Using Wolfram Mathematica, we managed to symbolically solve both linear programs for $r=1$ and a large enough number of nodes so that the maximum success probability is strictly less than $1$. Here we also present those results and their implications.

\subsection{Optimal colorings for rings of up to 11 nodes}
\label{sec:nonsig11nodeRing}

Let $\Phi$ be an arbitrary non-signaling coloring beyond distance $r$, and, for every coloring $\phi\in\Sigma^\Nodes$, let $p_\phi$ be the probability that $\Phi=\phi$. Letting 
$p_\phi\ge 0$ be variables that sum up to $1$, the maximum success probability for 3-coloring among non-signaling colorings is given by the linear program 
\begin{subequations}
\label{eqn:LPFull}
\begin{alignat}{2}
&\text{maximize}&\quad&\sum_{\text{proper }\phi\in\Sigma^\Nodes}p_\phi \\
&\text{subject to}&& p_\phi=p_{\phi\circ\mathsf{Prev}}  \qquad \text{for all }\phi\in\Sigma^\Nodes, 
\label{eqn:LPFullCirc}
\\
&&& 
\sum_{\substack{\phi\in\Sigma^\Nodes \\ \forall j\colon \phi(\omega_j+\fr{s_j})=\zeta_j}} p_\phi
= \sum_{\substack{\phi\in\Sigma^\Nodes \\ \forall j\colon \phi(\omega'_j+\fr{s_j})=\zeta_j}} p_\phi
\qquad \text{for all }F,\zeta,\omega,\omega',
\label{eqn:LPFullFrames}
\end{alignat}
\end{subequations}
where the latter constraint is for all collections of gap-$r$-placable sliding frames $F=(\fr{s_1},\ldots,\fr{s_t})$, all corresponding collections of tableaux $\zeta=(\zeta_1,\ldots,\zeta_t)$, and all pairs of gap-$r$ placements $\omega,\omega'$ of $F$.
The linear program (\ref{eqn:LPFull}) is feasible for all $r$; for example, consider $p_\phi=1/|\Sigma|^n$ for all $\phi$.

We symbolically solved (\ref{eqn:LPFull}) for $r=1$ and $n$ up to $11$.\footnote{\label{fn:adhoc}%
Our solution of the linear program was somewhat ad hoc. First, we solved the primal and the dual problems numerically. Then, from those numerical results, we managed to get symbolic feasible solutions for the primal and the dual with matching objective values, implying their optimality.} While for $n$ up to $9$ we got the optimal value of $1$, for $n=10$ and $n=11$ we got $2/3$ and $32/63$, respectively.
This constitutes the proof that single-round one-way quantum distributed algorithms cannot succeed with 3-coloring 10-node and 11-node rings with probabilities higher than 2/3 and 32/63, respectively. Note that, however, this by itself does not yet rule out the possibility that for a larger number of nodes perfect coloring is again possible. For that, we have to show a similar result for line segments.

\subsection{Impossibility of perfectly coloring the 9-node line segment}
\label{sec:nonsig9nodeLine}

We can consider the equivalent of (\ref{eqn:LPFull}) for $k$-node line segments. Let us consider interval $\fixFr{w}{w+k-1}$, where $w$, the number of the initial node, is irrelevant. 
The maximum success probability among random colorings of this line segment that are non-signaling beyond distance $r$ can be computed via the linear program
\begin{subequations}
\label{eqn:LPFullLine}
\begin{alignat}{2}
&\text{maximize}&\quad&\sum_{\text{proper }\psi\in\Sigma^k}p_\psi \\
&\text{subject to}&& 
\sum_{\substack{\psi\in\Sigma^k \\ \forall j\colon \psi(\omega_j+\fr{s_j})=\zeta_j}} p_\psi
= \sum_{\substack{\psi\in\Sigma^k \\ \forall j\colon \psi(\omega'_j+\fr{s_j})=\zeta_j}} p_\psi
\qquad \text{for all }F,\zeta,\omega,\omega',
\label{eqn:LPFullLineFrames}
\end{alignat}
\end{subequations}
where the maximization is over probability distributions $p=(p_\psi)_{\psi\in\Sigma^k}$ 
and where the constraint (\ref{eqn:LPFullLineFrames}) is for all collections of sliding frames $F=(\fr{s_1},\ldots,\fr{s_t})$, all corresponding collections of tableaux $\zeta=(\zeta_1,\ldots,\zeta_t)$, and all pairs of gap-$r$ placements $\omega,\omega'$ of $F$ that respect  $\fixFr{w}{w+k-1}$.

Compared to the linear program (\ref{eqn:LPFull}) for the whole ring, now, in (\ref{eqn:LPFullLine}), we do not have the symmetry introduced by the cyclicity condition (\ref{eqn:LPFullCirc}). This results in a less constrained and, for the same number of nodes, slower-to-compute program.
Because of the results in Section~\ref{sec:nonsig11nodeRing} and the following claim, for $r=1$ and for $k$ up to $8$, we do not need any computation at all: the optimum is $1$.

\begin{clm}
\label{clm:RingsToSegment.}
Suppose there is a random coloring of the $n$-node ring that is non-signaling beyond distance $r$ and whose success probability is $q$. Then its restriction to any line segment of at most $n-r$ nodes is a random coloring of that line segment that is non-signaling beyond distance $r$ and has the success probability at least $q$.
\end{clm}

For $r=1$ and $k=9$, we symbolically solved the linear program (\ref{eqn:LPFullLine}) and found the objective value to be $11/15$ (Footnote~\ref{fn:adhoc} applies here too). Thus, as a result, since non-signaling colorings are independent, Lemma~\ref{lem:SegmentsToExp} applies, and we obtain Theorem~\ref{thm:expBound}.

\bigskip

A couple of remarks are in place. First, we remark that our current results do not imply that the maximum among success probabilities of random colorings of the $n$-node ring that are non-signaling beyond distance $1$ is exponentially small in $n$. Currently, we have such a result only for independent colorings.

Second, while we have shown that, for $r=1$ and $k=9$, the success probability of (\ref{eqn:LPFullLine}) is strictly below $1$, it might be the case that, for any number of rounds $r\ge1$, there is a length $k$ such that the optimal solution of (\ref{eqn:LPFullLine}) is strictly below $1$. Finding such a solution numerically seems to be far beyond what is computationally possible. But the existence of such a solution for every $r$ would imply an $\omega(1)$ lower bound on the round complexity of $3$-coloring via quantum distributed algorithms. 

\section{Impossibility of perfectly coloring the 4-node ring}
\label{sec:independent:n4}

Let us consider random colorings of the ring that are non-signaling beyond distance $1$. If,
in addition to the non-signaling constraints, we impose that the colorings must be independent, we can show stronger no-go results.
In particular, while we managed to find a perfect non-signaling coloring for every $n$ up to $9$ (i.e., find a solution of (\ref{eqn:LPFull}) of objective value $1$), here we show that a prefect independent coloring does not exist already for a much smaller number of nodes $n$.

\begin{thm}
\label{thm:n4independent}
There is no perfect random $3$-coloring of the $4$-node ring that is independent beyond distance $1$. 
\end{thm}

\begin{proof}
Let us assume the contrary: there is such a coloring $\Phi$. 
There are $a_4=18$ proper $3$-colorings $\phi$ of the $4$-node ring, which we write below as $(\phi(0),\phi(1),\phi(2),\phi(3))\in\Sigma^4$ where we take $\Sigma=\{0,1,2\}$ to be the set of colors.
Since $\Phi$ is cyclic, there are probabilities $p_{0,1},p_{0,2},p_{1,2},q_0,q_1,q_2$ such that $\Phi$ equals each 
 \begin{alignat*}{2}
& (0, 1, 0, 1), \; (1, 0, 1, 0) &\qquad&\text{w.p. }p_{0,1}, \\
& (0, 2, 0, 2), \; (2, 0, 2, 0) &&\text{w.p. }p_{0,2}, \\
& (1, 2, 1, 2), \; (2, 1, 2, 1) &&\text{w.p. }p_{1,2}, \\
& (0, 1, 0, 2), \; (1, 0, 2, 0), \; (0, 2, 0, 1), \;  (2, 0, 1, 0) &&\text{w.p. }q_0, \\
& (1, 0, 1, 2), \;  (0, 1, 2, 1), \; (1, 2, 1, 0), \; (2, 1, 0, 1) &&\text{w.p. }q_1, \\
& (2, 0, 2, 1), \; (0, 2, 1, 2), \; (2, 1, 2, 0), \; (1, 2, 0, 2) &&\text{w.p. }q_2.
\end{alignat*}
We have $2(p_{0,1}+p_{0,2}+p_{1,2})+4(q_0+q_1+q_2)=1$.

Given a node $v$ and a color $\sigma\in\Sigma$, let $r_\sigma$ be the probability that $\Phi(v)=\sigma$. Note that $r_\sigma$ is independent from the choice of $v$ due to the cyclicity of $\Phi$. We have
$
r_0 = p_{0,1}+p_{0,2} + 2 q_0 + q_1 + q_2,
$
with analogous expressions for $r_1$ and $r_2$. From these three expressions, we obtain
\[
p_{0,1}  = (r_0+r_1-r_2)/2-q_0-q_1.
\]
Given a node $v$, the probability that both $\Phi(v)=1$ and $\Phi(v+2)=2$ is $q_0$. On the other hand, since $\Phi$ is independent beyond distance $1$, $\Phi(v)$ and $\Phi(v+2)$ are independent, taking colors $1$ and $2$ with probability $r_1$ and $r_2$, respectively. As a result, $q_0=r_1r_2$, and similarly $q_1=r_0r_2$.
Because $r_0+r_1=1-r_2$, we get $q_0+q_1=r_2-r_2^2$ and, further on, 
\[
p_{0,1}
  = (1-r_2-r_2)/2-r_2+r_2^2
 = \big(r_2-1-1/\sqrt{2}\big)\big(r_2-1+1/\sqrt{2}\big).
\]
For $p_{0,1}$ and $r_2$ both to be probabilities, we must have $r_2\le 1-1/\sqrt{2}<1/3$. The same argument yields $r_0<1/3$ and $r_1<1/3$, which is a contradiction because $r_0+r_1+r_2=1$.
\end{proof}

Theorem~\ref{thm:n4independent} implies that already for $n=4$, a single-round one-way quantum distributed algorithm cannot $3$-color perfectly. However, aside from stating that the maximum success probability cannot be $1$, the theorem does not place any other upper bound on the probability.
While, for non-signaling colorings, the maximum success probability can be characterized by a linear program, it is likely not the case for independent colorings, making its computation more difficult.

In the next section, we briefly discuss a more powerful model of quantum distributed algorithms whose round complexity can be bounded by the maximum success probability among non-signaling colorings, but not that among independent colorings.

\section{Discussion of Future Work}
\label{sec:FutureWork}

\paragraph{Exponential bounds given unique indices.}

In this paper we have considered the model of quantum distributed algorithms where initially each node's memory is in the same state. One could also consider a more powerful model where, at the start of the computation, each node $v\in\Nodes$ is given a unique label $i\in\bZ_n$. Here $v$ and $i$ do not need to be equal: $v$ lets us reason about node's location on the ring, but it is unknown to the computation, while $i$ can be used by the computation, but it does not reveal anything about labels of other nodes in its vicinity. Let us suppose here that these labels are assigned uniformly at random, there being $n!$ different assignments in total. In addition, we could assume that there is an $n$-partite quantum state (generalizing an $n$-partite random bitstring of the classical case) and that at the beginning of the computation the node labeled $i$ receives $i$-th part of this state.

Even if we give these extra resources to quantum $r$-round one-way distributed algorithms, along the same lines as for Lemma~\ref{lem:QuantumNonSignal}, it can be shown that the resulting random coloring $\Phi$ of the ring is non-signaling beyond distance $r$. Hence, the impossibility of perfect $3$-coloring in a single-round proven for $n\ge 10$ in Section~\ref{sec:SuccessProbLP} still applies to this stronger model. However, unlike for independent colorings, it is not clear if Lemma~\ref{lem:SegmentsToExp} applies to non-signaling colorings, and therefore it is not clear if the impossibility to perfectly color constant-length line segments can be elevated to exponential bounds for coloring the whole ring, as is done for Theorem~\ref{thm:expBound}. Showing that the success probability of quantum single-round one-way distributed algorithms to 3-color when these algorithms are aided by unique node labels and correspondingly shared quantum state is exponentially small is one problem left opened by the current work.

Better understanding the differences between non-signaling and independent colorings  might provide a better understanding of capabilities of quantum distributed algorithms with and without shared quantum states.

\paragraph{Generalization of round elimination.}

The round elimination technique~\cite{Brand2019RoundElem}, when applied to prove the $\Omega(\log^*\!n)$-round lower bound for $3$-coloring the ring by classical randomized distributed algorithms, is based essentially on the following observation.
Informally speaking, if there is an $r$-round randomized distributed algorithm that with high probability properly $c$-colors a line segment of length $r+1$ (i.e., $(r+2)$-node line segment), then there is an $(r-1)$-round randomized distributed algorithm that with high probability properly $2^c$-colors a line segment of length $r$. By repeatedly applying this observation, one finally reaches the ``base'' case, where one reasons about the probability of properly coloring two-node line segment without any communication.

This argument does not seem to work for non-signaling colorings because it reasons about line segments whose length is only $r+1$, yet the limitations of non-signaling colorings are due to long distance correlations.
It would be interesting to see if a similar technique to round elimination could be used for non-signaling colorings. In particular, one could potentially seek a statement similar to the following.

\begin{hyp}[informal]
There exists an integer constant $\delta\ge 1$, representing the number of rounds to be eliminated in one step, such that the following holds.
Suppose there is a line segment $\fixFr{1}{k}$, a random $c$-coloring $\Psi$ of this line segment that is proper with high probability, a frame-length $s$, and a round-number $r$  such that, for all gap-$r$ placements $\omega$ of the collection of frames $F=(\fr{s},\fr{s})$ that respect the line segment $\fixFr{1}{k}$, the marginal distributions over $(s,s)$-collections of tableaux induced by $\Psi(\omega+F)$ are the same.
Then there exists a random $2^c$-coloring $\Psi'$ of the same line segment such that $\Psi'$ is proper with high probability and, for all gap-$(r-1)$ placements $\omega'$ of the collection of frames $F'=(\fr{s-\delta},\fr{s-\delta})$ that respect the line segment $\fixFr{1}{k}$, the marginal distributions over $(s-\delta,s-\delta)$-collections of tableaux induced by $\Psi'(\omega'+F')$ are the same.
\end{hyp}

If the above hypothesis (or one similar to it) can be proven, we suspect that it might lead to a proof that, if $\Phi$ is a random 3-coloring of an $n$-node ring that is both non-signaling beyond distance $r$ and that is proper with high probability, then $r=\Omega(\log^*\!n)$. This would imply the same lower bound on the round complexity of quantum distributed algorithms.

Better understanding of the linear program (\ref{eqn:LPFullLine}) for $r=1$ and $k=9$ and, in particular, under what restricted sets of constraints its optimum is still strictly below $1$, might yield some insights towards the hypothesis above.

\section*{Acknowledgements}

The authors are grateful to Sebastian Brandt for useful insights. FLG and AR were supported by JSPS KAKENHI Grant No.~JP20H05966 and MEXT Quantum Leap Flagship Program (MEXT Q-LEAP) Grants Nos.~JPMXS0118067394 and JPMXS0120319794. FLG was also supported by JSPS KAKENHI Grants Nos.~JP19H04066, JP20H00579, JP20H04139, JP21H04879.

\appendix

\section{Quantum Distributed Algorithm}
\label{app:quantum}

We assume that the reader of this section is familiar with the basics of quantum computation (for an introductory textbook, see~\cite{NielsenChuang}).

\subsection{Model of quantum one-way distributed algorithm}
\label{sec:model}

Let the memory of a node $v\in\Nodes$ be held in register $\regN_v$, which consists of two subregisters: the workspace register $\regW_v$ and the (forward) message register $\regMn_v$. Let $\cH_{\regW_v}$, $\cH_{\regMn_v}$, and $\cH_{\regN_v}=\cH_{\regW_v}\otimes\cH_{\regMn_v}$ be complex Euclidean spaces corresponding to the registers $\regW_v$, $\regMn_v$, and $\regN_v$, respectively. 
We assume that each $\regW_0,\regW_1,\regW_2,\ldots$ consists of the same number of qubits, thus the dimensions of all $\cH_{\regW_0},\cH_{\regW_1},\cH_{\regW_2},\ldots$ are the same, and similarly we assume that each $\regMn_0,\regMn_1,\regMn_2,\ldots$  consists of the same number of qubits.
We additionally assume that all nodes run the same local computation.
We may drop the subscript $v$ when convenient.

The collective memory of nodes is initialized to the state $|\pmb{0}\>^{\otimes n}$, which is a pure state on $\bigotimes_{x\in[n]}\cH_{\regN_v}$; here $|\pmb{0}\>$ is the state of the register $\regN_v$ with all individual qubits set to $|0\>$.
Let $V$ be the \emph{messaging unitary} defined by the linear extension of its action on product states as 
\[
V \colon \bigotimes_{v\in \Nodes}|b_v,c_v\>_{\regN_v}
\mapsto
\bigotimes_{v\in \Nodes}|b_v,c_{v-1}\>_{\regN_v},
\]
where $|b_v\>\in\cH_{\regW_v}$ and $|c_v\>\in\cH_{\regMn_v}$. The unitary $V$ implements one round of communication.

The local actions of every node between two communication rounds are described by a unitary $U$, which acts on $\cH_{\regN}$. It is the same for every node and, without loss of generality, the same for every round.%
\footnote{We can assume that the same unitary $U$ is applied at every round because $U$ can be designed so that the node keeps a round counter in its memory and then controls its actions based on it.}
 Likewise, the final local projective measurement $\Pi=\{\Pi_\sigma\colon \sigma\in \Sigma\}$ is the same for every node. We permit this measurement not to be in the standard computational basis, therefore there is no loss of generality in not applying any local unitary after the last round of communication $V$.
 
An $r$-round protocol starts with the state $|\pmb{0}\>^{\otimes n}$ and then alternates between applications of $U$ for every node and $V$ on the joint system. Thus, the overall final state right before the final measurement is $(VU^{\otimes n})^r|\pmb{0}\>^{\otimes n}$.
Finally, each node performs the measurement $\Pi$ on its memory, outputting the result of the measurement. The joint result being random, we treat it as a random variable $\Phi$ taking values in $\Sigma^\Nodes$.

\subsection{Proof of Lemma~\ref{lem:QuantumNonSignal}}
\label{sec:quantumIndependence}

Suppose we have a quantum $r$-round one-way distributed algorithm specified by $U$ and $\Pi$.
Let the random coloring $\Phi$ be the output of the algorithm.

Consider $s\ge 1$ and $\zeta=(\sigma_0,\ldots,\sigma_{s-1})\in\Sigma^s$. Let us consider the space $\bigotimes_{i=-r}^{s-1}\cH_{\regN_i}$, where we label each individual instance of the space $\cH_{\regN}$ not with the number of a node but with an index $i$ in the gapped sliding frame $\gapFr{r}{s}$.
Let us define an orthogonal projector $\Xi_r(s,\zeta)$ recursively for $r\ge 0$ as follows. 
Intuitively, we will start with a final measurement operator $\Xi_0(s,\zeta)$ and then run the computation in reverse for $r$ rounds, seeing how that operator evolves under this reverse computation.
For $r=0$, let 
\[
\Xi_0(s,\zeta):=\bigotimes_{i=0}^{s-1}(\Pi_{\sigma_i})_{\regN_i}.
\]
For $r\ge 1$, first, let
\[
\Nodes_{r+s}
\colon
\bigotimes_{i=-r}^{s-1}\cH_{\regN_i}
\rightarrow
\cH_{\regW_{-r}}\otimes\bigotimes_{i=-r+1}^{s-1}\cH_{\regN_i} \otimes \cH_{\regMn_s}
\]
be the linear isometry defined as
\[
\Nodes_{r+s}
\colon
\bigotimes_{i=-r}^{s-1}|b_i,c_i\>_{\regN_i}
\mapsto
|b_{-r}\>_{\regW_{-r}}\otimes\bigotimes_{i=-r+1}^{s-1}|b_i,c_{i-1}\>_{\regN_i} \otimes |c_{s-1}\>_{\regMn_s}.
\]
Then, for $r\ge 1$, we recursively define
\[
\Xi_r(s,\zeta) :=
(U^{-1})^{\otimes(r+s)} \Nodes_{r+s}^{-1}
\Big(I_{\regW_{-r}}\otimes \Xi_{r-1}(s,\zeta)\otimes I_{\regMn_s}\Big)
\Nodes_{r+s} U^{\otimes(r+s)},
\]
where $I$ denotes the identity operator on the corresponding registers.

Given a fixed $s$, projectors $\Xi_r(s,\zeta)$ are orthogonal for distinct $\zeta$ and the sum $\sum_{\zeta\in\Sigma^{s}}\Xi_r(s,\zeta)$ equals the identity operator on $\bigotimes_{i=-r}^{s-1}\cH_{\regN_i}$.
Define the probability distribution $p_s$ over $\Sigma^s$ as
\[
p_s(\zeta) :=
\<\iniSt|^{\otimes(r+s)} \Xi_r(s,\zeta) |\iniSt\>^{\otimes(r+s)}.
\]

Now, given a node $v$, we can consider $\Xi_r(s,\zeta)$ to act on the space corresponding to the interval of nodes $v+\gapFr{r}{s}=\fixFr{v-r}{v+s-1}$ by associating $\cH_{\regN_i}$, where $i\in\gapFr{r}{s}$ is an index, with $\cH_{\regN_{v+i}}$, where $v+i$ is the number of a node. To specify that $\Xi_r(s,\zeta)$ acts on the space corresponding to those nodes, we write it as $(\Xi_r(s,\zeta))_{\regN_{v-r}..\regN_{v+s-1}}$.

\bigskip

Consider a collection of sliding frames $F=(\fr{s_1},\ldots,\fr{s_t})$ and, for every $j\in\{1,\ldots,t\}$, a tableau $\zeta_j\in\Sigma^{s_j}$.
Suppose $\omega$ is a gap-$r$ placement of $F$.
Because of the constraints on gap-$r$ placements, we can see that, for distinct $j$, the projectors  
$(\Xi_r(s_j,\zeta_j))_{\regN_{\omega_j-r}..\regN_{\omega_j+s_j-1}}$ act on spaces corresponding to disjoint sets of nodes.

Let use fix $\omega$ and $\zeta$, and let us write $\zeta_j=(\zeta_{j,0},\ldots,\zeta_{j,s_j-1})$ for every $j$.
Let us define
\[
p:=\Pr[\Phi(\omega+F)=(\zeta_1,\ldots,\zeta_t)],
\]
that is, $p$ is the probability that for every $j\in\{1,\ldots,t\}$ and every $i\in\{0,\ldots,s_j-1\}$ the node $\omega_j+i$ outputs $\zeta_{j,i}$. This probability equals
\[
p=\|\Xi_{\mathsf{final}}(VU^{\otimes n})^r|\iniSt\>^{\otimes n}\|^2,
\]
where
\[
\Xi_{\mathsf{final}} :=  \bigotimes_{j=1}^t\bigotimes_{i=0}^{s_j-1}(\Pi_{\zeta_{j,i}})_{\regN_{\omega_j+i}} \otimes I_{\text{rest}}
 =  \bigotimes_{j=1}^t(\Xi_0(s_j,\zeta_j))_{\regN_{\omega_j}..\regN_{\omega_j+s_j-1}} \otimes I_{\text{rest}}
\]
and $I_{\text{rest}}$ is the identity on the joint space of the nodes not in any $\omega_j+\fr{s_j}$.

Let us start with the final measurement operator $\Xi_{\mathsf{final}}$ and effectively run the computation in reverse, obtaining the measurement operator
\[
\Xi_{\mathsf{init}}:= (VU^{\otimes n})^{-r} \Xi_{\mathsf{final}}(VU^{\otimes n})^r
\]
on the initial state. We can write
\[
\Xi_{\mathsf{init}} = 
\bigotimes_{j=1}^t(\Xi_r(s_j,\zeta_j))_{\regN_{\omega_j-r}..\regN_{\omega_j+s_j-1}} \otimes I_{\text{rest}'},
\]
where $I_{\text{rest}'}$ is the identity on the joint space of the nodes not in any $\omega_j+\gapFr{r}{s_j}$.
The correctness of the above expression for $\Xi_{\mathsf{init}}$ can be seen inductively, using the recursive definition of projectors $\Xi_r$.

Finally, since the initial state is a product state, we get the desired independence. More precisely, 
because $p=\|\Xi_{\mathsf{init}}|\iniSt\>^{\otimes n}\|^2$, we get that
\(
p= p_{s_1}(\zeta_1)\cdot \ldots \cdot p_{s_t}(\zeta_t).
\)

\section{Bounds from Inspecting Pairs of Nodes}
\label{sec:No3Coloring}

Here we consider single-round algorithms, and we consider probability distributions over pairs of non-adjacent nodes, particularly focussing on whether the two nodes of the pair have the same color. We will be interested in pairs $(p,p')$ of such probability distributions, and we will be able to find pairs that serve as witnesses that one cannot $3$-color perfectly.

In Section~\ref{sec:boundsN11N22}, we consider non-signaling colorings of the $n$-node ring and computationally obtain witnesses for $n=11$ and $n=13,14,\ldots,22$ that bound the success probability strictly below $1$. In Section~\ref{sec:ExpInstance}, we obtain a similar witness for the $15$-node line segment, showing that the probability of properly coloring this segment is less than $\frac{1382}{1383}$, which then, using Lemma~\ref{lem:SegmentsToExp}, yields the following theorem, which is a weaker version of Theorem~\ref{thm:expBound}. 

\begin{thm}
\label{thm:expBoundSimple}
The probability that a single-round one-way quantum distributed algorithm properly 3-colors an $n$-node ring is at most $(\frac{1382}{1383})^{\lfloor n/16\rfloor}$.
\end{thm}

Both sections take advantage of the fact that, since $r=1$, the probability of two non-adjacent nodes having the same color is independent from the distance between them.
Results presented here might be of an independent interest in combinatorics.

\subsection{Constant bounds for certain small $n$}
\label{sec:boundsN11N22}

Suppose \(\Phi\) is a random coloring that is non-signaling at distances beyond one.
Let us consider the collection of frames $F=(\fr{1},\fr{1})$. There are $n(n-3)$ gap-one placements $\omega=(v,v+d)$ of $F$, which we specify by a node $v\in\Nodes$ and a distance $d\in\{2,3,\ldots,n-2\}$. Since $\Phi$ is non-signaling, the marginal distribution induced by $\Phi(\omega+F)=((\Phi(v)),(\Phi(v+d)))$ is the same for all such $\omega$.
In particular, the average probability
\[
\frac{1}{n}
\sum_{v\in \Nodes}\Pr[\sigma_1=\sigma_2\colon \sigma_1=\Phi(v) \And \sigma_2=\Phi(v+d)]
\]
is some constant $\gamma$ independent from $d$.

For every coloring $\phi\in\Sigma^\Nodes$, define $1_{\phi(v)=\phi(v+d)}$ to be $1$ when $\phi(v)=\phi(v+d)$ and $0$ otherwise. Then define
\[
\beta_{\phi,d}:=\frac{1}{n}\sum_{v\in\Nodes} 1_{\phi(v)=\phi(v+d)}
\in \{0,\tfrac{1}{n},\tfrac{2}{n},\ldots,1\},
\] 
and note that $\beta_{\phi,d}$ is the probability that two randomly chosen nodes distance $d$ apart have the same color. Clearly $\beta_{\phi,d}=\beta_{\phi,n-d}$. 
Let us also define
\(
\beta_\phi:=(\beta_{\phi,2},\beta_{\phi,3}\ldots,\beta_{\phi,n-2}).
\)
We have
\(
\mathbb{E}[\beta_{\Phi}]=(\gamma,\gamma,\ldots,\gamma).
\)

For any probability distribution over distances $p=(p_2,p_3,\ldots,p_{n-2})$ we have
\[
\gamma = p\cdot \mathbb{E}[\beta_\Phi] = \mathbb{E}[p\cdot\beta_\Phi].
\]
Consider two such probability distributions $p,p'$ and a \emph{bias} $\Delta \ge 0$ such that $p\cdot\beta_\phi \ge \Delta + p' \cdot\beta_\phi$ for all proper colorings $\phi$.
Note that such $p,p',\Delta$ always exist, as $p=p'$ and $\Delta=0$ jointly satisfy the required condition. The interesting cases, however, are when $\Delta>0$. The maximum value of $\Delta$ can be computed by the linear program
\begin{subequations}
\label{eqn:LP}
\begin{alignat}{2}
& \text{maximize} & \qquad& \Delta \\
& \text{subject to} & & (p-p')\cdot \beta_\phi \ge \Delta \qquad \text{for all proper }\phi,
\end{alignat}
\end{subequations}
where the optimization is over $\Delta\ge0$ and all probability distributions $p$ and $p'$. Since $\beta_{\phi,d}=\beta_{\phi,n-d}$, for the sake of aesthetics, without loss of generality we impose $p_d=p_{n-d}$ and $p'_d=p'_{n-d}$.%
\footnote{The shapes of plots in Figure~\ref{fig:distr}, representing optimal solutions of (\ref{eqn:LP}), are smoother this way. We could alternatively impose $p_d=p'_d=0$ for $d>\lfloor n/2 \rfloor$, but then, for even $n$, the optimal solutions we found would have $p'_{n/2-1}\approx 2p_{n/2}$. We note that, however, when we implement (\ref{eqn:LP}) in Wolfram Mathematica, we do indeed assume $p_d=p'_d=0$ for $d>\lfloor n/2 \rfloor$, because that approximately halves the number of variables.}
Note that, whenever the maximum bias $\Delta$ is non-zero, probability distributions $p$ and $p'$ achieving it have disjoint supports.

We solved the linear program (\ref{eqn:LP}) for $n$ up to $22$ and we found the optimal $\Delta$ to be strictly positive for $n=11$ and $n=13,14,\ldots,22$. Before we elaborate on these solutions, let us see how a non-zero $\Delta$  implies upper bounds on the success probability.

 Let $-\Gamma\le 0$ be the minimum of $(p-p')\cdot \beta_\phi$ over improper colorings $\phi$, and let $\varepsilon_\Phi:=\Pr[\Phi\text{ is improper}]$.
We have
\begin{align}
\notag
0 & = \mathbb{E}[(p-p')\cdot\beta_\Phi] 
\\ \notag
 & = \sum_{\text{proper }\phi\in\Sigma^\Nodes}\Pr[\Phi=\phi]\,(p-p')\cdot\beta_\phi
+ \sum_{\text{improper }\phi\in\Sigma^\Nodes}\Pr[\Phi=\phi]\,(p-p')\cdot\beta_\phi
\\ \label{eq:DeltaGammaErr}
& \ge   (1-\varepsilon_\Phi)\Delta - \varepsilon_\Phi\Gamma.
\end{align}
Hence, since $\Gamma\le 1$, we get
\[
\varepsilon_\Phi\ge \frac{\Delta}{\Delta+\Gamma} \approx \frac{\Delta}{\Gamma}\ge \Delta,
\]
where the approximation assumes $\Delta\ll\Gamma$, which will be the case for the instances we consider.

\paragraph{Solving the linear program.}

We used Wolfram Mathematica to solve the linear program (\ref{eqn:LP}) exactly (i.e., not in the floating-point arithmetic, but symbolically). The code for this is relatively simple, only slightly building on top of the code given in Appendix~\ref{app:code11}. We have made the code available online~\cite{code3Color}.

To find the set of all distinct $\beta_\phi$, we run a brute force search over all proper $\phi\in\Sigma^n$. For example, for $n=11$, $n=16$, $n=21$, the number of such distinct $\beta_\phi$ is, respectively, 21, 410, 8336. 

Once we had found optimal $p,p'$, for values of $n$ up to $19$, we run the brute force search over all improper colorings $\phi$ to compute $\Gamma=-\min_\phi\,(p-p')\cdot\beta_\phi$ for them. Since there are approximately $2^n$ proper colorings and approximately $3^n$ improper colorings, finding $\Gamma$ is computationally slower than finding $\Delta$.

Figure~\ref{fig:bias} shows lower bounds on the error probability of $3$-coloring the ring using quantum single-round one-way distributed algorithms for $n=11,12,\ldots,22$. For $n=12$, this lower bound is $0$.

\begin{figure}[!h]
\begin{centering}
\includegraphics[width=8cm]{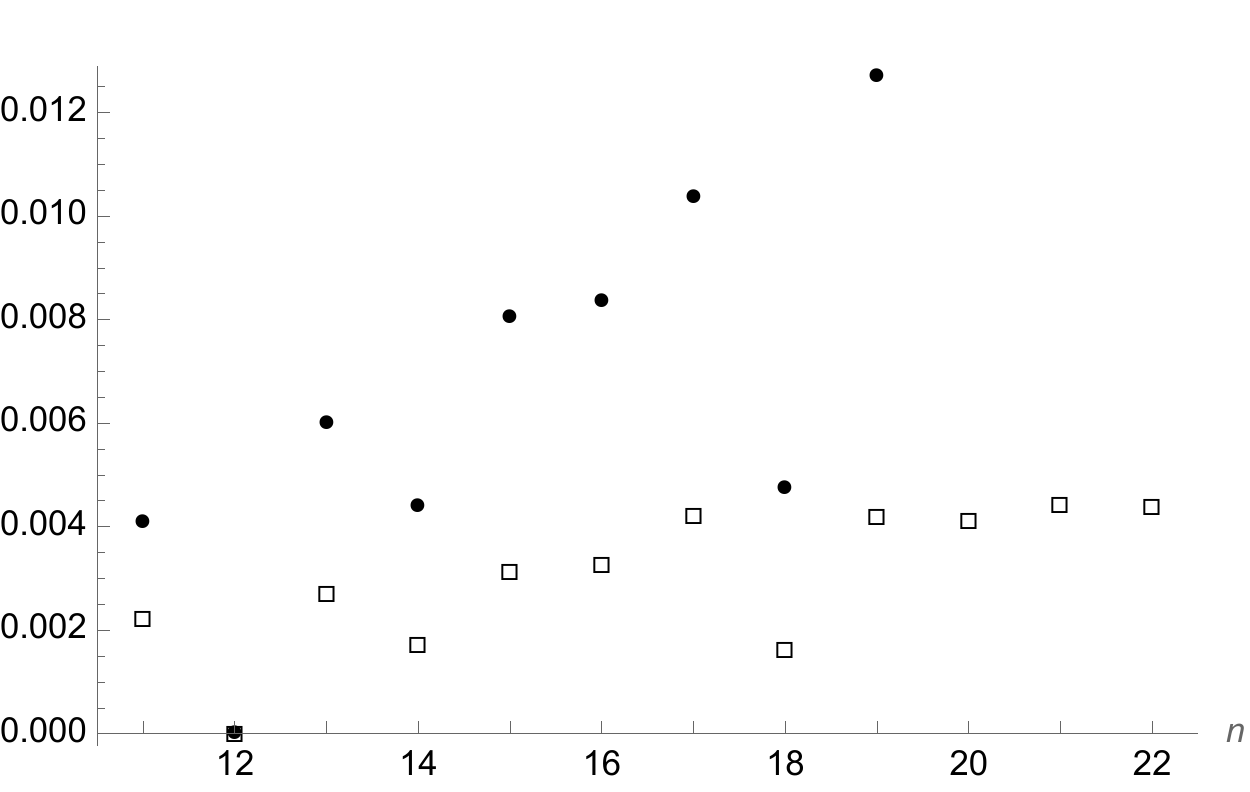}
\caption{\small
Maximum bias $\Delta$ for $n=11,12,\ldots,22$, represented by squares, as well as lower bounds $\frac{\Delta}{\Delta+\Gamma}$ on the error probability for $n$ up to $19$, represented by circles. (Since $\Gamma\le 1$, for $n=20,21,22$, we can take $\frac{\Delta}{\Delta+1}\approx\Delta$ as the lower bound on the error probability.)}
\label{fig:bias}
\end{centering}
\end{figure}

Figure~\ref{fig:distr} shows optimal solutions $p,p'$ of the linear program (\ref{eqn:LP}). We can see that $p$ is sharply concentrated around small distances (treating distances $d$ and $n-d$ as the minimum between them), while $p'$ is more flatly supported on long distances. In Figure~\ref{fig:distr} we have plotted $p,p'$ for $n=13,14,15,20,21,22$, and similarly shaped plots were observed for other values of $n$ between $11$ and $22$.

\begin{figure}[!h]
\begin{centering}
\includegraphics[width=15cm]{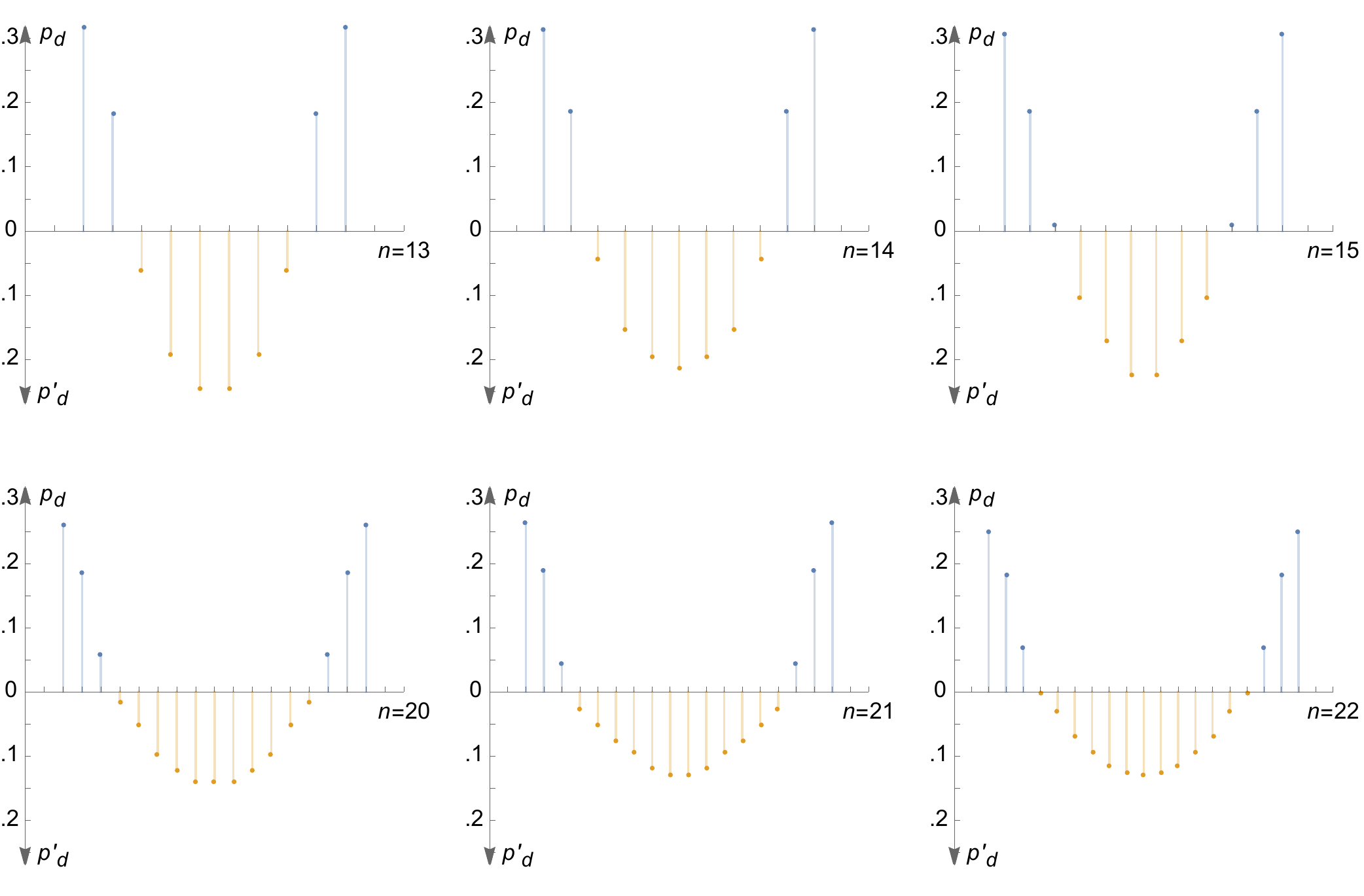}
\caption{\small
Optimal solutions $p,p'$ of the linear program (\ref{eqn:LP}) for $n=13,14,15,20,21,22$. Probability distributions $p,p'$ have disjoint supports, and we display only non-zero probabilities in them. The horizontal axis corresponds to the distance $d$. The upward-pointing vertical axis corresponds to $p_d$, while the downward-pointing vertical axis corresponds to $p'_d$.}
\label{fig:distr}
\end{centering}
\end{figure}

We note that, if a non-zero bias could be achieved by some $p$ and $p'$ that are supported on $\{r+1,r+2,\ldots,n-r-1\}$, then that would imply the impossibility of perfect 3-coloring in $r$ one-way rounds. However, we succeeded to find such $p$ and $p'$ only for $r=1$.

\subsection{Exponential bounds via Lemma~\ref{lem:SegmentsToExp}}
\label{sec:ExpInstance}

In a similarly fashion to the ring in Section~\ref{sec:boundsN11N22}, let us now consider a line segment $\fixFr{w}{w+k-1}$ of $k$ nodes. For notational convenience, let $w=1$. Let $\Psi$ be a random coloring of this segment that is independent beyond distance $1$. As before, let us consider the collection of frames $F=(\fr{1},\fr{1})$. There are $(k-1)(k-2)$ gap-$1$ placements $\omega=(u,v)$ of $F$ that respect  $\fixFr{1}{k}$, but, due to symmetry, it suffices to consider only the half of them that have $u<v$.

Similarly as before, for $u,v\in\{1,2,\ldots,k\}$ such that $v\ge u+2$ and a coloring $\psi\in\Sigma^k$ of the line segment, let $\beta_{\psi,u,v}=1$ if $\psi(u)=\psi(v)$, and $\beta_{\psi,u,v}=0$ otherwise, and we extend this definition to random $\Psi$.
Since $\Psi$ is non-signaling beyond distance $1$, $\Pr[\Psi(u)=\Psi(v)]$ is the same for all non-adjacent $u,v$, and let us denote this probability by $\gamma$. Let us consider the vectors $\beta_\psi:=(\beta_{\psi,u,v})_{u,v}$ and probability distributions $p:=(p_{u,v})_{u,v}$ over non-adjacent nodes expressed as vectors, where the indexing is done in the same consistent manner. We have
\(
\mathbb{E}[\beta_\Psi] = (\gamma,\gamma,\ldots,\gamma),
\)
and thus 
\(
\mathbb{E}[(p-p')\cdot\beta_\Psi] = 0
\).

Similarly as before, we define the maximum bias $\Delta$ as the optimal value of the linear program
\begin{subequations}
\label{eqn:LPsegment}
\begin{alignat}{2}
& \text{maximize} & \qquad& \Delta \\
& \text{subject to} & & (p-p')\cdot \beta_\psi \ge \Delta \qquad \text{for all proper }\psi.
\end{alignat}
\end{subequations}
We note that solving the linear program (\ref{eqn:LPsegment}) is more costly than (\ref{eqn:LP}), given the number of nodes, $k$ and $n$ respectively, are equal. That is because (\ref{eqn:LPsegment}) does not exploit symmetries due to the cyclicity of the ring. In particular, up to color permutations, one can recover $\psi$ from $\beta_\psi$.

We solved the linear program (\ref{eqn:LPsegment}) for increasing large values of $k$, and the smallest $k$ for which the optimum $\Delta$ is strictly positive is $k=15$. For $k=15$, the maximum bias is 
\[
\Delta = \frac{569\,800\,825}{2\,362\,818\,191\,739} \approx 2.412\cdot 10^{-4}.
\]
An optimal solution $(p,p')$ corresponding to this $\Delta$ is illustrated in Figure~\ref{fig:ppGrid}. As in Figure~\ref{fig:distr} for the whole ring, we see that $p$ is supported on pairs $(u,v)$ with shorter distances between them, while $p'$ on pairs $(u,v)$ with larger distances between them.

\begin{figure}[!h]
\begin{centering}
\includegraphics[width=11cm]{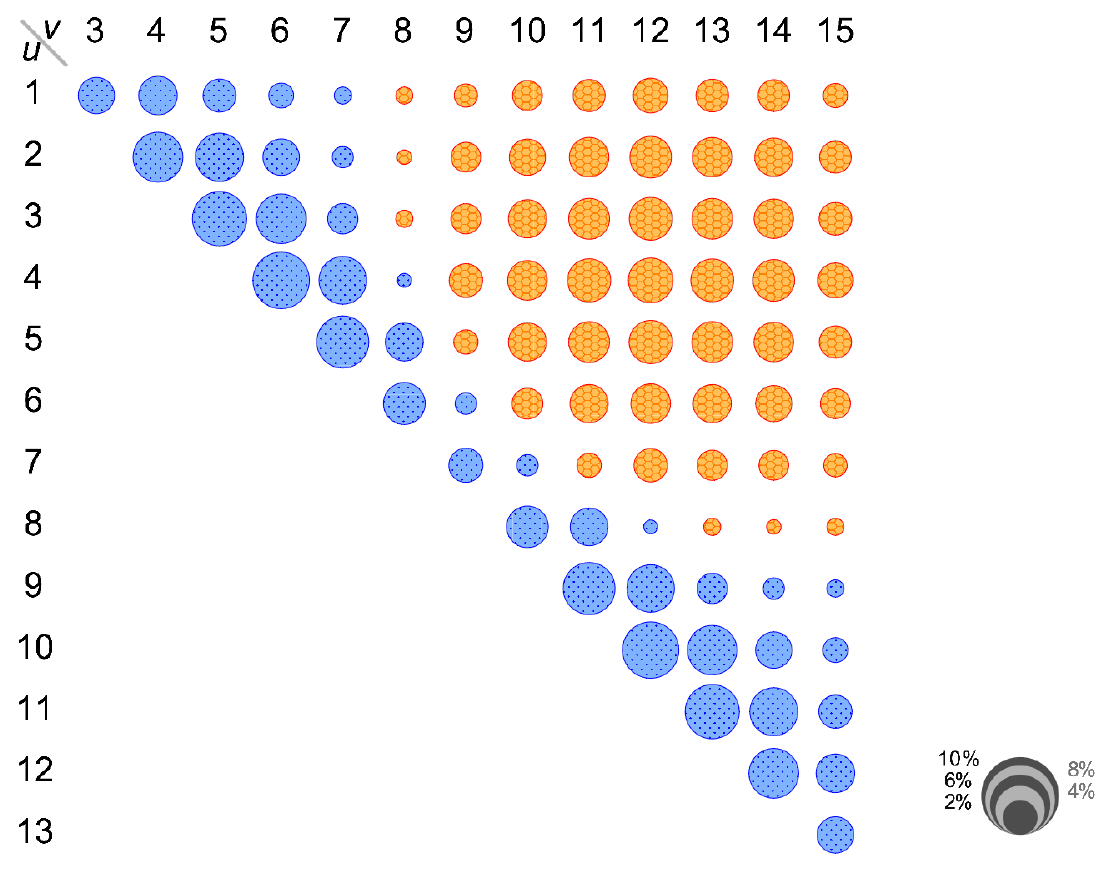}
\caption{\small
An optimal solution of (\ref{eqn:LPsegment}) for the $15$-node line segment. The blue dotted discs indicate non-zero $p_{u,v}$ and the orange honeycombed discs indicate non-zero $p'_{u,v}$. The probabilities are proportional to the area of the disc, the scale being indicated by the five overlapping discs in gray.
}
\label{fig:ppGrid}
\end{centering}
\end{figure}

Then, for this solution $(p,p')$, we computed the minimum $-\Gamma$ of $(p-p')\cdot\beta_\psi$ over all improper $\psi$, and we found
\[
\Gamma = \frac%
{218\,333\,768\,903\,290\,121}%
{655\,639\,517\,480\,121\,198}
\approx 0.3330
\]
(note that, since the solution $(p,p')$ is not necessarily unique, neither is $\Gamma$).
Just like the argument based on the inequality (\ref{eq:DeltaGammaErr}), this means that the success probability of $\Psi$ is at most $\frac{\Gamma}{\Delta+\Gamma}<\frac{1382}{1383}$. This together with Lemma~\ref{lem:SegmentsToExp} proves Theorem~\ref{thm:expBoundSimple}.

\section{Code for Witnessing Impossibility of Perfect Coloring}
\label{app:code11}

Here we consider the ring of $n=11$ nodes, and we present a Wolfram Mathematica code that verifies that, according to arguments in Section~\ref{sec:boundsN11N22}, it is impossible for a quantum single-round one-way distributed algorithm to perfectly $3$-color the ring.

Consider probability distributions $p=(\frac{30}{41},\frac{11}{41},0,0)$ and $p'=(0,0,\frac{14}{41},\frac{27}{41})$ over distances $2,3,4,5$, respectively. As in Section~\ref{sec:boundsN11N22}, for a coloring $\phi\in\Sigma^\Nodes$, let 
$\beta_\phi:=(\beta_{\phi,2},\beta_{\phi,3},\beta_{\phi,4},\beta_{\phi,5})$, where
$\beta_{\phi,d}:=|\{v\in\Nodes\colon\phi(v)=\phi(v+d)\}|/n$.
The following code returns $\texttt{minDelta}$ of $\frac{1}{451}$, showing that $(p-p')\cdot\beta_\phi\ge \frac{1}{451}$ for all proper colorings $\phi$.

\begin{Verbatim}[numbers=left,xleftmargin=6.5mm]
n = 11;
pp = {30, 11, -14, -27}/41;
minDelta = 1;
For[i = 0, i < 2^(n - 2), i++, 
  steps = 1 + IntegerDigits[i, 2, n - 2];
  phi = {1, 0}~Join~Mod[Total[Take[steps, #]] & /@ Range[n - 2], 3];
  If[Last[phi] == 1, Continue[]];
  betaN = {};
  For[d = 2, d <= Floor[n/2], d++, betaNd = 0;
   For[v = 1, v <= n, v++, 
    betaNd += Boole[phi[[v]] == phi[[Mod[v + d, n, 1]]]]];
   AppendTo[betaN, betaNd];];
  If[pp . betaN < n minDelta, minDelta = pp . betaN/n];];
Print[minDelta]
\end{Verbatim}

Line 1 of the code is self-explanatory. In line 2, we define constant \texttt{pp} as $p-p'$. In line 3, we introduce variable \texttt{minDelta} that will track the minimum $(p-p')\cdot\beta_\phi$ found so far, which is clearly no more than $1$. We consider the set of colors $\Sigma=\{0,1,2\}$. The \texttt{For} loop of lines 4--13 goes over all proper colorings of the ring that, without loss of generality, colors nodes $v=0$ and $v=1$ with colors $1$ and $0$, respectively. We note that not all counters $i$ correspond to proper colorings as about a third of them are discarded in line 7, which tests if $\phi(n-1)=1=\phi(0)$.

In line 5, we represent $i$ as an $n-2$ bit string, and then we add $1$ to all bits, getting a string $\mathtt{steps}\in\{1,2\}^{n-2}$. After that, in line 6, we effectively assign $\phi(j+1):=\phi(j)+\mathtt{steps}_j\mod 3$ for all $j$ from $1$ to $n-2$. Thus $\phi$ is a proper coloring aside from, potentially, having $\phi(n-1)=\phi(0)$, which we check in line 7.

In lines 8--12, we compute $\beta_\phi n$ in the variable \texttt{betaN}. Here \texttt{betaNd} contains $\beta_{\phi,d}n=|\{v\in\Nodes\colon\phi(v)=\phi(v+d)\}|$, which is simply counted in line 11. Note that Wolfram Mathematica indexes lists starting with index 1.
Finally, in line 13, we update the minimum $(p-p')\cdot\beta_\phi$ if necessary, and print the result in line 14.

\bigskip

To implement the linear program (\ref{eqn:LP}), one only has to slightly extend the code above. That code as well as all the other code that we have used for our computations is made available online~\cite{code3Color}.

\end{document}